\documentclass[letterpaper,11pt]{amsart}
\usepackage{times}
\usepackage[margin=1.00in]{geometry}
\usepackage{amsfonts,latexsym,bbm}
\usepackage[dvips]{epsfig}
\usepackage{graphicx}
\usepackage[linesnumbered,lined,ruled,noend,norelsize]{algorithm2e}
\providecommand{\DontPrintSemicolon}{\dontprintsemicolon}

\usepackage{enumitem}% http://ctan.org/pkg/enumitem

\newtheorem{theorem}{Theorem}%[section]
\newtheorem{blankthm}[theorem]{}

\newtheorem{lemma}[theorem]{Lemma}

\newtheorem{corollary}[theorem]{Corollary}

\newtheorem{prop}[theorem]{Proposition}

\newcommand{\suppress}[1]{}
\def\rset{\mathbb R}
\def\qed{\mbox{ }~\hfill~$\Box$}

\def\ep{\varepsilon}
\newcommand{\be}{\begin{equation}}
\newcommand{\ee}{\end{equation}}
\newcommand{\bea}{\begin{eqnarray}}
\newcommand{\eea}{\end{eqnarray}}
\newcommand{\bean}{\begin{eqnarray*}}
\newcommand{\eean}{\end{eqnarray*}}
\usepackage{color}

\def\y{{y}}

\def\x{{x}}

\def\z{{z}}

\newcommand{\aiout}{a_i^{\text{out}}}
\newcommand{\aiin}{a_i^{\text{in}}}
\newcommand{\Otilde}{\widetilde{O}}
\newcommand{\alphaiout}{\alpha_i^{\text{out}}}
\newcommand{\alphaiin}{\alpha_i^{\text{in}}}
\newcommand{\alphavout}{\alpha_v^{\text{out}}}
\newcommand{\alphavin}{\alpha_v^{\text{in}}}
\newcommand{\alphabar}{\bar{\alpha}}
\newcommand{\alphavmax}{\alpha_v^{\text{max}}}
\newcommand{\alphabarvout}{\alphabar_v^{\text{out}}}
\newcommand{\alphabarvin}{\alphabar_v^{\text{in}}}
\newcommand{\alphabarvmax}{\alphabar_v^{\text{max}}}

\newcommand{\rhoone}{\textstyle\sum_v \rho_v}

\newcommand{\normrowi}{\Vert a_{i\cdot}\Vert}
\newcommand{\normcoli}{\Vert a_{\cdot i}\Vert}

\begin{document}
\title{Analysis of a Classical Matrix Preconditioning Algorithm}
\author{Leonard J. Schulman}\thanks{LJS was supported in part by NSF grants 1038578 and
1319745 and by the Simons Institute for the Theory of Computing, where most of this work
was performed.  AS was supported in part by NSF grants 1016896 and 1420934.
Addresses: Caltech, Engineering and Applied Science MC305-16, Pasadena CA 91125, USA, {\tt schulman@caltech.edu}; Computer Science Division, Soda Hall, University
of California, Berkeley CA 94720--1776, USA, {\tt sinclair@cs.berkeley.edu}. }
\author{Alistair Sinclair}

%%%%%%%%%%
%%%%% abstract
%%%%%%%%%%

\begin{abstract}
We study a classical iterative algorithm for balancing matrices in the $L_\infty$ norm via a scaling transformation. This algorithm, which goes back to Osborne and Parlett \& Reinsch in the 1960s,
is implemented as a standard preconditioner in many numerical linear algebra packages.
Surprisingly, despite its widespread use over several decades, no bounds were known on
its rate of convergence.  In this paper we prove that, for any irreducible $n\times n$ 
(real or complex) input matrix~$A$, a natural variant of the algorithm converges in 
$O(n^3\log(n\rho/\varepsilon))$ elementary balancing operations, where $\rho$ measures
the initial imbalance of~$A$ and $\varepsilon$ is the target imbalance of the output matrix.  
(The imbalance of~$A$ is $\max_i |\log(\aiout/\aiin)|$, where $\aiout,\aiin$ are the maximum
entries in magnitude in the $i$th row and column respectively.)  This bound is tight
up to the $\log n$ factor.  A balancing operation
scales the $i$th row and column so that their maximum entries are equal, and requires
$O(m/n)$ arithmetic operations on average, where $m$ is the number of non-zero elements in~$A$.
Thus the running time of the iterative algorithm is $\Otilde(n^2m)$. This is the first time bound
of any kind on any variant of the Osborne-Parlett-Reinsch algorithm.
We also prove a conjecture of Chen that characterizes those matrices for which
the limit of the balancing process is independent of the order in which 
balancing operations are performed.
\end{abstract}
\newpage
\setcounter{page}{1}

\maketitle
\thispagestyle{empty}

%%%%%%%%%%
%%%%% intro
%%%%%%%%%%

\section{Introduction}\label{sec:intro}
\subsection{Background and discussion of results}\label{subsec:results}
In numerical linear algebra, it is standard practice to {\it precondition\/} an
$n\times n$ real or complex matrix~$A$ by performing a similarity transform 
$D^{-1}AD$ for some diagonal scaling matrix~$D$, prior to performing other computations
on~$A$ such as computing its eigenvalues.  The goal is that $D^{-1}AD$ should
be {\it balanced}, in the sense that the norm of the $i$th row is equal to the norm
of the $i$th column, for all~$i$.   The point here is that standard
linear algebra algorithms tend to be numerically unstable for unbalanced matrices.
Diagonal scaling achieves balance without affecting the eigenvalues of~$A$.
(Preconditioning typically also involves a separate process of row and column permutations
which we ignore here.)

This idea goes back to Osborne in 1960~\cite{Osborne60}, who suggested an iterative
algorithm for finding a~$D$ that balances~$A$ in the $L_2$ norm, and proved that it converges in the limit; he also proposed using the analogous iteration in the $L_\infty$ norm. 
Parlett and Reinsch generalized the algorithm to other norms~\cite{ParlettR69}.  
Neither of these papers gave
any bound on the convergence time of the algorithm (in any norm).
In the decades since then, this algorithm
has been implemented as standard in almost all numerical linear algebra software, 
including EISPACK, LAPACK and MATLAB. For further background
see, e.g.,~\cite{TrefethenEmbree05,Kressner05}. Surprisingly,
despite its widespread use in practice, no bounds are
known on the running time of any variant of this iterative balancing algorithm.

Our goal in this paper is to initiate a quantitative study of the Osborne-Parlett-Reinsch
algorithm. We emphasize that our motivation is to understand an existing method which
has emerged as the leading choice of practitioners, rather than to devise a new competitor
in the asymptotic regime; however, the bounds we obtain show that even asymptotically
this method is not far off the theoretically best (and much more complex) algorithms.
(See the Related Work section for a discussion.)

The iterative algorithm is very easy to describe, and involves repeated execution of a
simple local balancing operation.   
Let $\| \cdot \|$ be a norm. 
For an index $i\in [n]$, let $\normrowi$
and $\normcoli$ denote the norms of the $i$th row and $i$th column of~$A$ respectively.
To avoid technicalities, we make the standard assumption that $A$ is
irreducible (see below), which ensures these norms are always nonzero; otherwise,
earlier (and faster) stages of the preconditioning decompose~$A$ into irreducible
components.  A {\it balancing operation at}~$i$ scales each entry in column~$i$ of~$A$ by 
$r_i={\small\sqrt{\normrowi/\normcoli}}$, and each entry in row~$i$ by~$r_i^{-1}$
(thus leaving $a_{ii}$ unchanged).
Note that this operation corresponds to the diagonal
transformation $D^{-1}AD$, where $D={\rm diag}(1,\ldots,r_i,\ldots,1)$.  The iterative algorithm simply performs the following step repeatedly:\footnote{The algorithm 
in \cite{Osborne60,ParlettR69} operates on the indices in a fixed cyclic order.  For ease
of analysis we pick indices randomly.  We conjecture that this change makes little difference
to the rate of convergence.}
 \textit{Pick an index $i$ and perform a balancing operation at~$i$.}

We focus in this paper on the case of balancing in the $L_\infty$ norm, where the goal
is to find~$D$ so that in $D^{-1}AD$ the largest entry (in magnitude) in the $i$th row is equal to
the largest entry in the $i$th column, for all~$i$.  We refer to such a matrix as {\it balanced}. 
The scaling factor in the algorithm is then
just $r_i={\small\sqrt{\aiout/\aiin}}$, where $\aiout=\max_j |a_{ij}|$ and $\aiin = \max_j |a_{ji}|$ are 
the maximum elements in the $i$th row and column respectively.  This version of the algorithm
is particularly simple to implement, which makes it an attractive alternative to (say)
the $L_2$ version.  Since the algorithm depends only on the magnitudes of the $a_{ij}$, 
and multiplies them only by positive reals, we will simplify from now on
by assuming that all entries $a_{ij}$ are non-negative, keeping in mind that
the balancing operations are actually performed on the original 
matrix~$A$.\footnote{It is sometimes also assumed
that the diagonal entries of~$A$ are replaced by zeros before the balancing process, since balancing never alters them. 
This replacement may of course change the problem since (e.g.) a diagonally dominant matrix
is already balanced in~$L_\infty$. We do not assume such a replacement has been made; our analysis of the balancing algorithm applies whether or not diagonal entries are zero.}

It is instructive to view this procedure in terms of the directed weighted graph~$G_A$, 
in which there is an edge $(i,j)$ of weight $a_{ij}$ if $a_{ij}>0$ (and no edge if $a_{ij}=0$).
Irreducibility of~$A$ corresponds to $G_A$ being strongly connected. The above balancing 
operation scales all incoming (respectively, outgoing) edge weights at vertex~$i$ by 
$r_i={\small\sqrt{\aiout/\aiin}}$ (respectively, $r_i^{-1}$),
where $\aiout$, $\aiin$ are respectively the maximum outgoing
and incoming weights at~$i$.  At first sight (taking logarithms of the edge weights)
this process resembles a diffusion
on~$G_A$; however, the fact that the scaling at each step depends
on the {\it maximum\/} incoming/outgoing edge weight is a nonlinearity
that makes the process much harder to analyze.

Clearly any fixed point of the above iteration must be a balanced matrix~$A$,
i.e., $\aiout=\aiin$ for all~$i$; in other words, in the weighted graph~$G_A$,
the maximum incoming edge weight is equal to the maximum outgoing edge weight
at every vertex. Note that the product of the edge weights along any cycle in~$G_A$ is an
invariant of the algorithm. It is thus
natural to call two matrices {\it equivalent\/} if they agree on all these invariants.
The fact that the algorithm always converges asymptotically to a balanced matrix
was proved, surprisingly recently, by Chen~\cite{Chen-MS98}\footnote{To make
the current paper self-contained, we present in Appendix~\ref{app:yi-reproof}
an alternative proof of convergence that uses combinatorial rather than topological
techniques and is closer to the methodology developed in this paper.}; in particular
this implies the existence of a balanced matrix in each equivalence class. 
Chen also conjectured
that a worst-case input (in terms of rate of convergence) is one in which $G_A$ is
simply a directed cycle; in this case each vertex has only one incoming and one
outgoing edge, so the balancing operation is linear and convergence is easily seen
to occur within $\Theta(n^3)$ operations.
(The $\Theta$ here hides factors that depend on the
maximum imbalance in the input matrix and the desired bound on how close the output
matrix is to being balanced.)  In most cases the process
seems to converge much faster, but the problem of analyzing the convergence
rate in general graphs, or even in any strongly connected graph other than a directed cycle,
has remained open until now.  
For example, even the case of two directed cycles that share a common vertex
is quite non-trivial and has proved resistant to standard arguments.

The first difficulty one faces in analyzing convergence rates is that the balanced
matrix to which the algorithm converges is in general not uniquely defined, but may
depend on the sequence of indices~$i$ selected.  
\begin{figure}
$$A=
\left(
\begin{array}{cccc}
0&2&0&0\\
8&0&2&0\\
0&1&0&2\\
0&0&8&0
\end{array}
\right);\qquad
B_1=
\left(
\begin{array}{cccc}
0&4&0&0\\
4&0&2&0\\
0&1&0&4\\
0&0&4&0
\end{array}
\right);\qquad
B_2=
\left(
\begin{array}{cccc}
0&4&0&0\\
4&0&1&0\\
0&2&0&4\\
0&0&4&0
\end{array}
\right)
$$
\caption{\it Example of a non-UB matrix~$A$.  The balanced matrix~$B_1$ results
from $A$ by balancing at indices~1,4, while $B_2$ results by balancing at indices 2,4.}
\label{fig:simpleex}
\end{figure}
(See Figure~\ref{fig:simpleex} for a simple example that illustrates this phenomenon.)
Our first result proves a conjecture of Chen~\cite{Chen-MS98}, which 
characterizes all cases in which the balancing problem has a unique solution.
For a balanced matrix $B=\{b_{ij}\}$ and any real $w>0$, let $G_B^w$ 
denote the subgraph of~$G_B$ consisting of those edges of magnitude at least~$w$,
with isolated vertices removed. (A vertex with a self-loop is not considered isolated.)
\begin{theorem}\label{thm:characterizeUB}
A balanced matrix~$B$ is the unique balanced matrix in its equivalence class
if and only if  $G_B^w$ is strongly connected for all~$w$.
\end{theorem}
Note that this result gives an implicit criterion for whether a given input 
matrix~$A$ can be uniquely balanced: namely, that it is equivalent to a balanced
matrix~$B$ with the above property.  We say that such matrices~$A$ satisfy
the {\it Unique Balance (UB)\/} condition.

In an earlier version of the present paper~\cite{SSSTOC}, we focused on the UB case
and showed that a certain variant of the Osborne-Parlett-Reinsch algorithm converges
in $O(n^3\log n)$ balancing operations on any UB input matrix~$A$.  In this paper
we extend that analysis to all input matrices, for a slightly different variant of the
algorithm.  (The term ``variant" here refers only to the order in which balancing operations
are performed.)    

To describe this variant, we need the notion of ``raising" and ``lowering" operations,
which are one-sided versions of the standard balancing operation.
A {\it raising operation\/} at~$i$ performs a standard balancing
operation at~$i$ if $\aiout>\aiin$ and does nothing otherwise; similarly, a
{\it lowering operation\/} at~$i$ performs a standard balancing operation at~$i$
if $\aiout<\aiin$ and does nothing otherwise. Our variant of the algorithm will
perform a sequence of random raising operations followed by a sequence of random
lowering operations; equivalently, it can be viewed as the standard algorithm (with
a random index order) in which some of the balancing operations are censored
according to a very simple condition.

To specify the algorithm precisely, we also need a measure of how far a matrix is from
balanced.  We thus define the \textit{imbalance} of~$A$ as
$\max_i |\log(\aiout/\aiin)|$, and say that $A$ is {\it $\varepsilon$-balanced\/} if its imbalance is
at most~$\varepsilon$.  The target balance parameter $\varepsilon>0$ is provided to
the algorithm as an additional input.  We also need one-sided versions of these
definitions: namely, $A$ is {\it $\varepsilon$-raising-balanced\/} if 
$\max_i \log(\aiout/\aiin)\le\varepsilon$ (or \textit{raising-balanced} if $\ep=0$), and {\it $\varepsilon$-lowering-balanced\/} if 
$\max_i \log(\aiin/\aiout)\le\varepsilon$ (or \textit{lowering-balanced} if $\ep=0$).  Plainly $A$ is $\varepsilon$-balanced iff
it is both $\varepsilon$-raising-balanced and $\varepsilon$-lowering-balanced.

\begin{algorithm}
\DontPrintSemicolon
\KwIn{An $n\times n$ matrix $A$ with (unknown) imbalance $\rho$;
and a target $\ep>0$. }
\KwOut{A matrix equivalent to $A$ that is $\varepsilon$-balanced with high probability.}
Compute $\rho$ in time $O(n^2)$ \;
Repeat $T=O(n^3\log(\rho/\varepsilon))$ times\hskip1.0in {\it //Raising Phase//}\;
\hskip0.25in Pick an index $i$ u.a.r.\ and apply a raising operation at~$i$\;
Repeat $T=O(n^3\log(\rho/\varepsilon))$ times\hskip1.0in {\it //Lowering Phase//}\;
\hskip0.25in Pick an index $i$ u.a.r.\ and apply a lowering operation at~$i$\;
Output the resulting matrix\;
\caption{2PhaseBalance$(A,\varepsilon)$\label{alg3}}
\end{algorithm}

We note that even the asymptotic convergence of this variant of the algorithm is
not immediately clear, and does not follow from Chen's proof~\cite{Chen-MS98} or 
from our proof in Appendix~\ref{app:yi-reproof}. The reason is that those proofs
assume that the sequence of balancing operations is {\it fair}, in the sense that
each index~$i$ appears arbitrarily often.  In the two-phase version, fairness
cannot be guaranteed since some balancing operations are censored. Instead, convergence
of the two-phase algorithm is an immediate consequence of the following theorem.
\begin{theorem}\label{thm:newconv}
For any input matrix~$A$, any fair sequence of raising (resp., lowering)
operations converges to a unique raising-balanced (resp., lowering-balanced) matrix,
independent of the order of the operations.
\end{theorem}
Correctness of the two-phase algorithm follows because it is easy to
see that the lowering phase cannot increase the raising imbalance of~$A$ 
(see Proposition~\ref{prop:twophasebasic}).  The value~$T$ is chosen to
guarantee that $A$ is $\varepsilon$-raising-balanced (resp., $\varepsilon$-lowering-balanced)
at the end of a phase w.h.p.; hence, at the conclusion of both phases $A$ will indeed be 
$\varepsilon$-balanced w.h.p.

Theorem~\ref{thm:newconv} is actually a somewhat surprising result: even when $A$ is not UB
(so that the balanced matrix to which the original version of the algorithm converges
depends on the sequence of operations performed), if we perform only raising
(or only lowering) operations we are guaranteed a unique limit.  It is this fact that
will allow us to obtain an analysis of the rate of convergence.

We are now able to state the main result of the paper, which 
bounds the rate of convergence of the above algorithm.
\begin{theorem}\label{thm:main}
On input $(A,\varepsilon)$, where $A$ is an arbitrary irreducible non-negative matrix
with imbalance~$\rho$, the above two-phase $L_\infty$ balancing algorithm
performs $O(n^3\log(n\rho/\varepsilon))$ balancing operations and outputs
a matrix equivalent to~$A$ that is $\varepsilon$-balanced w.h.p.\footnote{We
use the phrase {``}with high probability (w.h.p.)'' to mean with probability tending to~1 as
$n\to\infty$.}
\end{theorem}

This is the first time bound (except on the cycle) for any variant of the
Osborne-Parlett-Reinsch algorithm (under any norm). Moreover, 
the bound on the convergence rate is actually tight up to a factor $O(\log n)$,
in view of the $\Omega(n^3)$ lower bound we mentioned earlier for the cycle.
Indeed, Theorem~\ref{thm:main} implies that the cycle is a worst case input for the algorithm
(up to the $\log n$ factor), as conjectured by Chen~\cite{Chen-MS98}.

{\bf Remark.}
Throughout the paper, we quantify the rate of convergence of the algorithm
in terms of the number of balancing operations performed.  Since each balancing operation
(raising or lowering)
involves finding a maximum element in one row and column and then scaling the row and
column, it can be performed on average in $O(m/n)$ arithmetic operations, 
where $m$ is the number of edges in~$G_A$ (i.e., the
number of non-zero entries in~$A$).  Thus, in terms of arithmetic operations,
the (average) running time of the algorithm is $O(n^2m \log n)$.

We close this section with a brief outline of our approach to analyzing the algorithm.
As we have observed above, it suffices to analyze only the raising phase; by
symmetry, an identical analysis applies to the lowering phase\footnote{It is an interesting
open question whether the introduction of phases in the Osborne-Parlett-Reinsch algorithm
leads to more rapid convergence in practice.}.
The first key observation is that the uniqueness of the limit in the raising phase allows us
to assign a well-defined \textit{height} to every vertex in~$G_A$; the height of a vertex
measures the {``}amount of raising" that needs to be performed at that vertex in order to
reach the unique raising-balanced configuration.
Each raising operation can then be viewed as smoothing these heights locally. 
In a diffusion process this naturally leads to 
an associated Laplacian potential
function consisting of the sum of squares of local height differences, whose convergence
is captured by the eigenvalues of the heat kernel operator. Attempts to conduct a similar analysis here fail. 
One serious problem is that a balancing operation at vertex~$i$
depends only on the incoming and outgoing edges of {\it maximum\/} weight, but has side effects
on all other edges incident at~$i$.  Another problem is that there are, essentially, 
``levels'' within the graph: roughly speaking, vertices that lie in the ``lower'' levels of 
the graph cannot reliably converge toward (raising) balance until the vertices in the ``higher"
levels have converged.   While phenomena similar to this can arise
in diffusion processes, in our nonlinear setting standard tools such as eigenvalues are
lacking to capture them.

The chief technical challenge in our analysis is to relate the local height changes
achieved by raising operations to an improvement in a global potential function
(the sum of the heights), thus ensuring significant progress (in expectation) over time.
The details of our analysis ensure an expected improvement that is a factor
$\Omega(n^3)$ of the current value of the potential function in each step, thus 
leading to global convergence in $\tilde{O}(n^3)$ raising operations.  

To establish the above local-global connection, we need a
two-dimensional representation of the current state of the algorithm, which 
records not only the height of a vertex but also its {\it level}; unlike the heights,
which change over time, the level of a vertex is fixed.  
This two-dimensional representation leads naturally
to the notion of the {\it momentum\/} of a vertex, which measures how far
a vertex currently is from its true level.  Momentum is a key ingredient
in relating local improvements to the global potential function.

\subsection{Related work}\label{subsec:related}
As mentioned above, the idea of iterative diagonal balancing was introduced by
Osborne~\cite{Osborne60} in 1960, who was motivated by the observation that
minimizing the Frobenius norm of~$A$ is equivalent to balancing~$A$ in the $L_2$ norm.
Osborne formulated an $L_2$ version of the above iterative algorithm and proved that
it converges in the limit (but with no rate bound); he also proposed the $L_\infty$ algorithm discussed in this paper but
left the question of convergence open.  Convergence was first 
proved by Chen~\cite{Chen-MS98} almost 40 years later.  
Parlett and Reinsch~\cite{ParlettR69} generalized Osborne's 
algorithm to other norms (without proving convergence) and discussed a 
number of practical implementation issues for preconditioning. For the $L_1$ version of the Osborne-Parlett-Reinsch algorithm, convergence was proved by Grad~\cite{Grad71}, uniqueness of the
balanced matrix by Hartfiel~\cite{Hartfiel71}, and a characterization of it by Eaves 
et al.~\cite{EavesHRS85}, but again no bounds on the running time of $L_1$ balancing were given.
 
Despite the scarcity of theoretical support, these algorithms are implemented
as standard in many numerical linear algebra packages and experience shows
them to be both useful and fast.

A sequence of two papers offers an alternative, and substantially more complex, 
non-iterative algorithm for matrix balancing in the $L_\infty$  norm.
Schneider and Schneider~\cite{SchneiderS91} gave an $O(n^4)$-time algorithm based on finding maximum
mean-weight cycles in a graph; the running time was improved to $O(nm + n^2\log n)$
using Fibonnaci heaps and other techniques by Young, Tarjan and 
Orlin~\cite{YoungTarjanOrlin91}. This is asymptotically faster than the $\Otilde(n^2m)$
worst case running time we establish for the iterative algorithm;
and for some graphs (e.g., the cycle) it is faster than the actual running time of the iterative algorithm.

However the iterative method has been favored in practice. This may be justified by the empirical distribution of inputs.
Moreover, it is definitely driven by the fact that the iterative method offers steady partial progress, and so can deliver, without being run to completion, a matrix that is sufficiently balanced for the subsequent linear algebra computation. In practice indeed the method is usually run for far fewer iterations than are needed in the worst case. 
Our purpose in this paper is to provide the first theoretical understanding of the widely used iterative methods, rather than to derive new
theoretical bounds for the underlying problem. 

In other work on matrix balancing, Kalantari {\it et al.}~\cite{Kalantari97} considered balancing in the $L_1$ norm, and provided the first polynomial time
algorithm by reduction to convex programming: $\Otilde(n^4)$ via the ellipsoid method.

Diagonal scaling has also been used to minimize matrix norms without regard to
balancing.  For example, Str\"om~\cite{Strom72} considers the problem of finding
diagonal scaling matrices to minimize the max and Frobenius norms of the matrix.
In particular, for the max norm he proves that (when $A$ is irreducible)
the optimal diagonal scaling matrix is obtained from the principal eigenvector.
Chen and Demmel~\cite{ChenD00} show that a suitable notion of {\it weighted\/}
balancing can be used to minimize the 2-norm, and
discuss Krylov-based algorithms that work efficiently on sparse matrices.
Boyd {\it et al.}~\cite{Boyd94} formulate the problem of minimizing the Frobenius
norm as a generalized eigenvalue problem.

A different notion of matrix balancing is often known as {\it Sinkhorn balancing\/} after
Sinkhorn~\cite{Sinkhorn64}, who proposed a natural iterative algorithm analogous
to that of Osborne-Parlett-Reinsch.  The goal here is to find a scaling matrix~$D$
such that $D^{-1}AD$ has prescribed row and column sums.  A polynomial
time algorithm for this problem was given by Kalantari and Khachian~\cite{KalKhach96}, 
while the convergence rate of Sinkhorn's iterative algorithm was studied by several
authors \cite{FL89,KalKhach93,LSW00,KLRS08}.  In particular Linial, Samorodnitsky
and Wigderson~\cite{LSW00} gave strongly polynomial bounds (after a non-Sinkhorn
preprocessing step) and derived a surprising approximation scheme for the permanent
of a non-negative matrix.

%%%%%%%%%%
%%%%% prelim
%%%%%%%%%%
%\input{2prelim}
\section{Preliminaries} \label{sec:prelim}
In this section we introduce some terminology and notation that we will use
throughout the paper.  We begin with an equivalent reformulation of the 
Osborne-Parlett-Reinsch balancing algorithm from the Introduction.

Let $A$ be an irreducible $n\times n$ matrix.
As described in the Introduction, the algorithm
operates by iteratively picking an index~$i$ and geometrically averaging the maximum
entries (in magnitude) $\aiout$ and $\aiin$.  It is convenient to switch to arithmetic averages
by setting $\alpha_{ij}=\log |a_{ij}|$ (so that 0 entries of~$A$ become $-\infty$).  
Letting $\alphaiout=\max_j \alpha_{ij}$ and $\alphaiin = \max_j \alpha_{ji}$,
a {\it balancing operation at~$i$} then consists of adding $(\alphaiout-\alphaiin)/2$
to the $i$th column and subtracting the same quantity from the $i$th row.
(In the actual implementation of the algorithm, this corresponds to the multiplicative
diagonal transformation 
$D={\rm diag}(1,\ldots,r_i,\ldots,1)$ 
applied to the original matrix~$A$.  However, for the rest of the paper we will assume
that the algorithm operates additively on the logarithms of the matrix entries and ignore
the details of this trivial translation.) 

We re-use the notation $G_\alpha$ from the Introduction
to denote the digraph with edges $(i,j)$ such that $\alpha_{ij}> -\infty$.  
(The passage from roman to greek letters will always resolve this abuse of notation.)
Irreducibility of~$A$ again corresponds to $G_\alpha$ being
strongly connected.  We shall regard $\alpha$ as a real-valued function defined
only on the edges of~$G_\alpha$, and call it a \textit{graph function}.  Note that
if balancing operations change the graph function $\alpha$ to $\alpha'$ we
always have $G_\alpha=G_{\alpha'}$.  For the remainder of the paper, 
we use letters $u,v$ etc.\ rather than $i,j$ to denote vertices as our focus 
will be on the graph $G_\alpha$ rather than on the matrix~$A$.

We say that a graph function~$\alpha$ is {\it balanced\/} if $\alphavout=\alphavin$ for all~$v$.
For each~$v$, define $\rho_v = |\alphavout-\alphavin|$.  
The {\it imbalance\/} of~$\alpha$
is $\rho:=\max_v\rho_v$, and $\alpha$ is {\it $\varepsilon$-balanced\/} if $\rho\le\varepsilon$.
(These definitions are equivalent to those in the Introduction in terms of the matrix
entries~$a_{ij}$.)  A balancing operation is available at~$v$ iff $\rho_v>0$.  

As indicated
in the Introduction, we analyze a slightly modified version of the basic Osborne-Parlett-Reinsch
algorithm, which has two phases and performs only certain balancing operations (raising or lowering) during a phase.  
Let $\rho_v^R=\max\{0,\alphavout-\alphavin\}$ and $\rho_v^L=\max\{0,\alphavin-\alphavout\}$
be the {\it raising\/} and {\it lowering imbalances\/}, respectively, at~$v$.
Note that $\rho_v = \rho_v^R+\rho_v^L = \max\{\rho_v^R,\rho_v^L\}$, since one of
$\rho^R_v$, $\rho^L_v$ must be zero.
A {\it raising operation\/} at~$v$ performs a standard balancing operation at~$v$ if
$\rho^R_v>0$ and does nothing otherwise; similarly, a lowering operation at~$v$
performs a standard balancing operation at~$v$ if $\rho^L_v>0$ and does nothing otherwise.
Such an operation is said to {\it raise\/} (resp., {\it lower\/})~$v$ by an amount $\rho_v^R/2$
(resp., $\rho_v^L/2$); note that a raising operation at~$v$ reduces its raising imbalance
to zero unless the maximum outgoing edge from~$v$ is a self-loop, in which case the
raising imbalance is reduced by a factor of~2.
The {\it raising imbalance\/} of~$\alpha$ is $\rho^R:=\max_v\rho^R_v$, and the 
{\it lowering imbalance\/} of~$\alpha$ is $\rho^L:=\max_v\rho^L_v$.
We say that $\alpha$ is {\it $\varepsilon$-raising-balanced\/} if $\rho^R\le\varepsilon$, and
{\it $\varepsilon$-lowering-balanced\/} if $\rho^L\le\varepsilon$.  Clearly $\alpha$ is
$\varepsilon$-balanced iff it is both $\varepsilon$-raising-balanced and
$\varepsilon$-lowering-balanced.

For any real~$w$, denote by $G_\alpha^w$ the subgraph of~$G_\alpha$
consisting of edges $(u,v)$ such that $\alpha_{uv}\ge w$, with isolated vertices
removed. (A vertex with a self-loop is not considered isolated.)

Let $\alpha,\gamma$ be two graph functions on the same graph~$G$.  We say that
$\alpha,\gamma$ are {\it equivalent}, $\alpha\sim\gamma$, if their sums around all
cycles of~$G$ are equal, i.e., if  
$\sum_{\ell=1}^k (\alpha_{v_\ell,v_{\ell+1}}-\gamma_{v_\ell,v_{\ell+1}})=0$ 
for all cycles $(v_1,\ldots,v_k)$, where $v_{k+1}=v_1$.

It is convenient for our analysis to introduce a generalization of the standard local balancing
operation at a vertex.  Let $\alpha$ be a graph function.  For any $x \in \rset$ and $S \subseteq [n]$, 
denote by $\alpha+xS$ (``$\alpha$ shifted by $x$ at the set $S$'') the graph function
given by $(\alpha+xS)_{uv}=\alpha_{uv}+x \cdot (I_S(u)-I_S(v))$, where
$I_S$ is the indicator function for membership in~$S$.
Thus the balancing operation at~$v$ is equivalent to shifting by
$(\alphavin-\alphavout)/2$ at the set~$\{v\}$.
It is readily seen that two graph functions are equivalent iff there is a sequence of shifts
converting one to the other.  Note that, like balancing operations,
shifts of~$\alpha$ do not change~$G_\alpha$.

%%%%%%%%%%
%%%%% ub
%%%%%%%%%%
\section{The Unique Balance property} \label{sec:ub}

As shown by Chen (and reproved with a different argument in the present
paper---see Appendix~\ref{app:yi-reproof}),
balancing operations do not cycle; 
Chen argued this topologically while we show that there is an order on graph
functions such that nontrivial balancing operations are strictly lowering in the order. 
However, neither the topological nor the order-theoretic argument yields
useful quantitative information on the rate of convergence.
One reason for this, but not the only reason, is that the limit point of the process is not in general
uniquely determined by the input matrix~$A$ (see Figure~\ref{fig:simpleex}
for a simple example).

For a broad class of inputs, however, a unique limit does exist. The characterization
of this class was conjectured by Chen~\cite{Chen-MS98} and is our first result.
This was already stated as Theorem~\ref{thm:characterizeUB} in the Introduction; we restate
it here in slightly different notation.
\par\smallskip
{\bf Theorem~\ref{thm:characterizeUB}.}      
{\it 
A balanced graph function~$\beta$ is the unique balanced graph function in its
equivalence class if and only if  $G_\beta^w$ is strongly connected for all~$w$.}
\par\smallskip
{\it Proof of $\Rightarrow$}: 
Let $\beta$ be balanced.  Suppose $G_\beta^w$ is not strongly connected,
and that $w$ is the largest value for which this is the case. 
Note that since $\beta$ is balanced, $G_\beta^w$ has no sources
or sinks, and hence must contain at least one source strongly connected
component (scc)~$S$ and one sink scc~$T$, both of
which are non-trivial (i.e., contain at least two vertices or one self-loop).
Moreover, by maximality of~$w$,
the weights of external edges (between scc's) in $G_\beta^w$ are all exactly~$w$.
Choose vertices $s\in S$ and $t\in T$, and consider a directed simple
path~$P$ from $s$ to~$t$ in~$G_\beta$.
(Such a path exists because
$G_\beta$ is strongly connected.)  Let $\varepsilon>0$ be smaller than the gap
between any two distinct values of $\beta$. Then in $G_{\beta+\varepsilon T}$,
the length of~$P$ is $\varepsilon$ less than it is in~$G_\beta$.

Now apply balancing operations in any order to the graph function $\beta+\varepsilon T$ (which is  equivalent to~$\beta$), to balance it (in the
limit).
We claim that no balancing operation will ever be applied at any vertex inside
a non-trivial scc of~$G_\beta^w$, and hence at either~$s$ or~$t$.
To see this, note that such a balancing operation can only occur if the weight
of some external edge incident on the scc reaches a value larger than~$w$.
But in $\beta+\varepsilon T$ no external edge has weight larger than~$w$
(this was true in $\beta$ and the shift at~$T$ has increased the
weight only of edges leaving~$T$, but to a value less than~$w$), so balancing
operations at the vertices outside the scc's cannot increase the weight of edges above~$w$.
Hence in the limiting (and hence balanced) graph function, the length of the path~$P$
from~$s$ to~$t$ will remain $\varepsilon$ less than it is in~$\beta$. 
(The length of a path can be altered only by a balancing operation
at one of its endpoints.) Hence this limiting graph function is different from~$\beta$.

{\it Proof of $\Leftarrow$}:
Let $\alpha \sim \beta$ be distinct balanced graph functions.
We show that there is a $w$ such that neither 
$G_\beta^w$ nor $G_\alpha^w$ is strongly connected.

Let $w$ be the largest value such that $G_\beta^w\neq G_\alpha^w$
(the supremum is achieved as the graph is finite), and suppose
w.l.o.g.\ there is an edge $(u,v)$ in $G_\beta^w-G_\alpha^w$.
Then $G_\beta^w$ cannot be strongly connected because if it were,
consider any cycle in~$G_\beta^w$ which includes~$(u,v)$.
The weight of this cycle is larger in $G_\beta$ than in $G_\alpha$ because all weights
larger than~$w$ in the cycle are identical in both functions, and in~$G_\beta$ 
all remaining weights equal~$w$, while in $G_\alpha$ all are at most~$w$
and at least one is strictly less than~$w$. Hence $\alpha\not\sim\beta$,
which is a contradiction.

Likewise, if there is an edge in $G_\alpha^w-G_\beta^w$ then
$G_\alpha^w$ is not strongly connected. This leaves only the
case $G_\alpha^w \subset G_\beta^w$.
Let $U$ be the set of vertices from
which $u$ can be reached in $G_\beta^w$, and $V$ be the set of
vertices reachable from $v$ in $G_\beta^w$.
As shown in the previous paragraph, these sets are disjoint. Each set must contain
a non-trivial scc (because $\beta$ is balanced and therefore
$G_\beta^w$ contains no sources or sinks). These scc's (call
them $U'$ and $V'$) must also be strongly connected in $G_\alpha^w$; 
in fact, we argue that $\alpha$ equals $\beta$ on the edges of $G_\beta^w$
within each scc. For, consider any cycle using edges of $G_\beta^w$
in the scc. The cycle may contain edges heavier than $w$ (which are
identical in $\alpha$ and $\beta$), and some edges that are equal to~$w$
in $\beta$ and $\leq w$ in~$\alpha$. But then since $\alpha \sim \beta$, 
these edges must also equal $w$ in $\alpha$. So
$\alpha$ and $\beta$ are identical on the edges of $G_\beta^w$ in the scc. Now, if
$G_\alpha^w$ is strongly connected, then it contains a path
from $V'$ to $U'$. Then since $G_\alpha^w \subset G_\beta^w$,
this path also lies in~$G_\beta^w$. Contradiction.
\qed
\par\medskip
In light of Theorem~\ref{thm:characterizeUB}, 
we say that a graph function~$\alpha$ has the {\it Unique Balance (UB)\/}
property if there is a balanced $\beta \sim \alpha$ such that $G_\beta^w$ is strongly
connected for all~$w$.  By the convergence theorem (see Appendix~\ref{app:yi-reproof}),
on any input matrix with the UB property the original Osborne-Parlett-Reinsch balancing algorithm
will converge to a unique balanced matrix.  In an earlier version of the present
paper~\cite{SSSTOC}, we analyzed the rate of convergence of (a suitable variant of)
the algorithm in the UB case, and derived the same bound as we now obtain for general
matrices in Theorem~\ref{thm:main}.
In the remainder of this paper, we turn our attention to the general case and leave 
Theorem~\ref{thm:characterizeUB} as an interesting structural result.

%%%%%%%%%%
%\input{twophaseconv}
\section{Convergence of the two-phase algorithm}\label{sec:convergence}
%%%%%%%%%%
The goal of this section is to prove that the two-phase variant of the iterative
balancing algorithm described in the Introduction converges and outputs
a $\varepsilon$-balanced graph function.  Note that convergence does not
follow from the general convergence proof given in Appendix~\ref{app:yi-reproof},
since the sequence of balancing operations in this case is not ``fair" in the sense that
every vertex is guaranteed to appear arbitrarily often.

Convergence of the two-phase variant follows immediately from Theorem~\ref{thm:newconv}
stated in the Introduction, which we restate here in graph function terms.
\par\smallskip
{\bf Theorem~\ref{thm:newconv}.}
{\it 
For any initial graph function~$\alpha$, any fair sequence of raising (resp., lowering)
operations converges to a unique raising-balanced graph function~$\alpha^R$
(resp., lowering-balanced graph function~$\alpha^L$),
independent of the order of the operations.
}
\par\smallskip

To deduce that the output of the two-phase algorithm is indeed $\varepsilon$-balanced,
we need the following additional simple observation.
\begin{prop}\label{prop:twophasebasic}
A lowering operation
cannot increase any raising imbalance $\rho^R_v$.  Symmetrically, a
raising operation cannot increase any lowering imbalance~$\rho^L_v$.
\end{prop}
\begin{proof}
We give the proof only for lowering operations; the proof for raising operations is symmetrical.
Consider a lowering operation at vertex~$v$, and assume that it has a non-zero effect, i.e.,
$\rho^L_v>0$.  Lowering at~$v$ reduces $\rho^L_v$ (to zero unless the maximum incoming
edge at~$v$ is a self-loop in which case the reduction is to $\rho_v^L/2$) and keeps $\rho^R_v$ at zero.
The only other imbalances that can be affected are at neighbors of~$v$.  Lowering at~$v$
decreases edge weights $\alpha_{uv}$, which cannot increase~$\rho^R_u$.
Similarly, lowering at~$v$ increases edge weights $\alpha_{vw}$, which again cannot
increase~$\rho^R_w$.
\end{proof}

It is now easy to see that the two-phase algorithm terminates with an $\varepsilon$-balanced
graph function.  By Theorem~\ref{thm:newconv} each of the two phases terminates with 
the appropriate balance condition.  At the end of the raising phase the function is
$\varepsilon$-raising balanced, and by Proposition~\ref{prop:twophasebasic} it remains
so throughout the lowering phase.  At the end of the lowering phase the resulting
graph function is both $\varepsilon$-raising balanced and $\varepsilon$-lowering
balanced, and hence $\varepsilon$-balanced as required.

The remainder of this section is devoted to proving Theorem~\ref{thm:newconv}.  We focus
on raising operations; the proof for lowering operations follows by a symmetrical argument.

We introduce the following notation.  Fix the initial graph function~$\alpha$
and a particular fair infinite sequence of raising operations.  
Let $\alpha=\alpha(0), \alpha(1),\ldots,\alpha(t),\ldots$ be the resulting sequence of graph
functions obtained.  Also, for each $t\ge0$ let the {\it raising vector\/} $r(t)=(r_v(t))$ record
the cumulative amount by which each vertex has been raised.  I.e., $r(0)=0$ and, if the $t$th raising 
operation is at vertex~$v$, then $r_v(t) = r_v(t-1)+\rho^R_v/2$ and $r_u(t) = r_u(t-1)$ for
all $u\ne v$.  Note that the vector $r(t)$ completely specifies\footnote{However,
the converse is not quite true: given $\alpha(t)$ we can deduce $r(t)$ only up to an additive constant.} $\alpha(t)$ via the relation 
\begin{equation}\label{eqn:alphar}
   \alpha_{uv}(t) = \alpha_{uv}(0) + r_v(t) - r_u(t).
\end{equation}

\begin{lemma}\label{lem:rbounded}
The sequence $r(t)$ is increasing and bounded above, and hence converges to a limit
vector~$r^*$ as $t\to\infty$.
\end{lemma}
\begin{proof}
The fact that $r(t)$ is increasing is immediate from its definition.  To see that it is bounded,
observe first that the graph function $\alpha(t)$ is bounded above and below for all~$t$;
indeed, it is immediately clear from the definition of balancing that $\max_{u,v}\alpha_{uv}(t)$
cannot increase with~$t$, while the less obvious fact that $\alpha_{uv}(t)$ is bounded below 
follows from Lemma~\ref{lem:bdedbelow} in Appendix~\ref{app:yi-reproof}.

Now define the function $\Phi(t) = \sum_{(u,v)\in E, u\ne v} N_{uv}(t)! \alpha_{uv}(t)$, where $N_{uv}(t)$
is the number of edges $(u',v')$ of weight $\alpha_{u'v'}(t)\le\alpha_{uv}(t)$.
We claim that, if the $t$th raising operation is at vertex~$v$, the change in~$\Phi$ is
bounded by
\begin{equation}\label{eq:facbound}
   \Phi(t+1)-\Phi(t) \le -\rho^R_v/2.
\end{equation}
Since by the observations in the previous paragraph $\Phi(t)$ is always bounded, 
and the increase in $r_v(t)$ is exactly $\rho^R_v/2$,
we conclude from~(\ref{eq:facbound}) that $r_v(t)$ is also bounded.

To prove~(\ref{eq:facbound}), represent each edge $(u,v)$ by a particle on the real
line at position~$\alpha_{uv}(t)$.  We will view $[t,t+1]$ as a continuous time interval,
during which each particle $(u,v)$ moves at constant speed from $\alpha_{uv}(t)$
to $\alpha_{uv}(t+1)$.  We now compute the rate of change of~$\Phi$ under a 
raising operation at vertex~$v$ (assuming that $\rho^R_v>0$ so that the raising
operation is non-trivial; otherwise, $\Phi$ does not change).  The only particles
that move are $(u,v)$ and $(v,w)$ as $u,w$ range over neighbors of~$v$; each particle
$(u,v)$ moves a distance $\rho^R_v/2$ to the right, and each $(v,w)$ moves the same
distance to the left.
Let $\bar u, \bar w$ be such that $\alpha_{{\bar u}v} = \max_{(u,v)}\alpha_{uv}(t)$
and $\alpha_{v{\bar w}} = \max_{(v,w)}\alpha_{vw}(t)$, so that 
$\rho^R_v = \alpha_{v\bar w}(t) - \alpha_{\bar u v}(t)$.   
The rate of change of $\Phi(t')$ for $t'$ in the interval $[t,t+1)$ is given by
\begin{eqnarray*}
   \frac{\partial\Phi(t')}{\partial t'} &=& \frac{\rho^R_v}{2}\Bigl(\sum_{(u,v)} N_{uv}(t')! - \sum_{(v,w)} N_{vw}(t')!\Bigr) \cr
        &\le& \frac{\rho^R_v}{2}\Bigl((N_{v\bar w}(t') -1) (N_{v\bar w}(t')-1)! - N_{v\bar w}(t')!\Bigr) \cr
        &\le&  -\frac{\rho^R_v}{2}. 
\end{eqnarray*}
The negative term in the second line arises from the contribution of the particle $(v,\bar w)$,
whose edge remains heaviest throughout.  The positive term arises because there
are at most $N_{v\bar w}(t') -1$ particles to the left of $(v,\bar w)$, each of which
contributes at most $(N_{v\bar w}(t')-1)!$.

As the particles move and pass other particles, the coefficients $N_{uv}(t')$ change.
However, observe that the above calculation is valid at all times 
$t'\in [t,t+1)$, since it relies only on the fact that $(v,\bar w)$ is the rightmost moving
particle.  Hence we deduce the bound~(\ref{eq:facbound}) claimed earlier, which 
completes the proof.
\end{proof}

It follows immediately from Lemma~\ref{lem:rbounded} and equation~(\ref{eqn:alphar})
that $\alpha(t)$ also converges to a limit~$\alpha^*$, and  
by continuity and fairness $\alpha^*$ must be raising-balanced.  It remains to show that $\alpha^*$
is independent of the sequence of raising operations.  For this we require the following
monotonicity property of raising operations.

\begin{lemma}\label{lem:monotonic}
Let $r,s$ be raising vectors obtained from a common initial graph function~$\alpha$ via
different sequences of raising operations.  Suppose that $r\le s$, and let $r',s'$ be the new
raising vectors obtained from $r,s$ after one additional raising operation at a common
vertex~$v$.  Then $r'\le s'$.
\end{lemma}
\begin{proof}
Clearly we have $r'_u=r_u\le s_u=s'_u$ for all $u\ne v$, so we need only show that $r'_v\le s'_v$.
By definition of the raising operation at~$v$, we have (using equation~(\ref{eqn:alphar}))
\begin{eqnarray*}
r'_v &=& r_v + \frac{\rho_v^R}{2}\\
&=& r_v + \frac{1}{2} \max\bigl\{0,\max_{w: v\to w} (\alpha_{vw} +r_w -r_v) - \max_{u:u\to v} (\alpha_{uv} +r_v -r_u)\bigr\}\\
&=& \frac{1}{2} \max\bigl\{2r_v,\max_{w: v\to w} (\alpha_{vw} +r_w) + \min_{u:u\to v} (r_u - \alpha_{uv})\bigr\}.\\
\end{eqnarray*}
(Here, and subsequently, we use the notation $u\to v$ to indicate that there is an edge
$(u,v)$ in~$G_\alpha$.)
Since the right-hand side here is monotonically increasing in all entries of~$r$, and $r\le s$,
we conclude that $r'_v\le s'_v$ as required.
\end{proof}

\begin{lemma}\label{lem:twophaseconv}
For any initial graph function~$\alpha$ and sequence of raising vectors
$r(1), r(2), \ldots, r(t), \ldots$ obtained from a fair sequence of raising
operations, the limit $r^*=\lim_{t\to\infty}r(t)$ exists and depends only on~$\alpha$,
not on the sequence of raising operations.
\end{lemma}
\begin{proof}
Consider some fair sequence of raising operations starting from~$\alpha$, and let
$\alpha(t)$, $r(t)$ be the corresponding sequences of graph functions and raising vectors
respectively.  By Lemma~\ref{lem:rbounded}, the limits $\alpha^*=\lim_{t\to\infty}\alpha(t)$
and $r^*=\lim_{t\to\infty}r(t)$ exist and $\alpha^*$ is raising-balanced.  

Now consider a second sequence of raising operations
also starting from~$\alpha$, and let $s(t)$ be the corresponding
sequence of raising vectors.  Again, by Lemma~\ref{lem:rbounded}, the limit
$s^*=\lim_{t\to\infty} s(t)$ exists.

Suppose we apply this second sequence of raising operations, starting on the one hand
from~$\alpha$ and on the other hand from~$\alpha^*$.  The resulting raising vectors
(with respect to~$\alpha$)
are respectively~$s^*$ and~$r^*$ (the latter because $\alpha^*$ is already raising-balanced).
Applying Lemma~\ref{lem:monotonic} to compare the raising vectors at each step of this process
yields $s^* \le r^*$.  But by symmetry we must also have $r^* \le s^*$, and hence
$r^* = s^*$.  This completes the proof.
\end{proof}

Since the graph function~$\alpha^*$ is completely determined by the raising vector~$r^*$,
Lemma~\ref{lem:twophaseconv}, and its symmetric counterpart for lowering operations,
proves Theorem~\ref{thm:newconv}, which was the main goal of this section.  

Figure~\ref{fig:RLexample} shows an example of a graph function and its associated 
unique corresponding raising- and lowering-balanced functions.  

\begin{figure}[h]
\includegraphics[height=100mm]{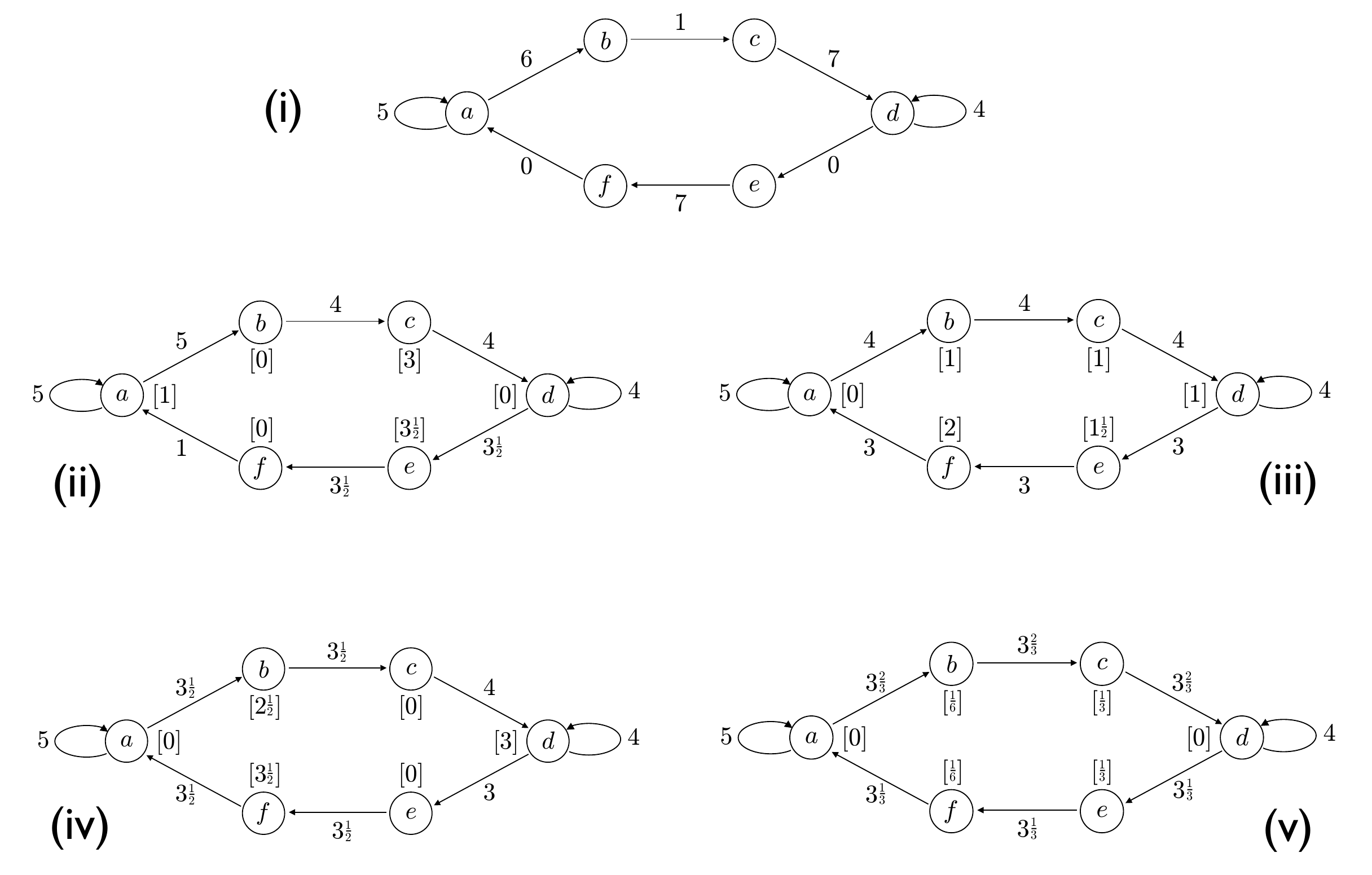}
\caption{\it Example showing convergence of raising and lowering operations.  (i)~A graph 
function~$\alpha$.  (ii)~The raising-balanced graph function~$\alpha^R$ obtained 
in the limit from~$\alpha$ by raising operations; numbers in square brackets at the
vertices show the limiting raising vector~$r^*$. 
(iii)~The balanced graph function~$(\alpha^R)^L$ obtained in the limit by applying
lowering operations to~$\alpha^R$, together with its associated lowering vector.  
(This is the unique balanced
graph function to which the two-phase algorithm ``converges" if both phases
are run to convergence.)
Figures~(iv) and~(v) show the graph functions~$\alpha^L$ and $(\alpha^L)^R$
respectively (and their associated lowering/raising vectors), 
obtained by reversing the order of the raising and lowering phases.
Note that the balanced graph functions $(\alpha^R)^L$ (in~(iii)) and 
$(\alpha^L)^R$ (in~(v)) are not the same; while both versions (raising-lowering and
lowering-raising) converge to a unique balanced graph function, these limits
are not equal.  The graph function~$\alpha$ is not UB.}
\label{fig:RLexample}
\end{figure}

%%%%%%%%%%
%\input{rate}
\section{Rate of convergence}\label{sec:rate}
%%%%%%%%
In this section we will prove the main result of the paper, Theorem~\ref{thm:main} from
the Introduction, which gives a tight bound on the rate of convergence of the iterative
balancing algorithm.  We restate the result here for convenience.
\par\smallskip
{\bf Theorem~\ref{thm:main}.}
{\it On input $(A,\varepsilon)$, where $A$ is an arbitrary irreducible non-negative matrix
with imbalance~$\rho$, the above two-phase $L_\infty$ balancing algorithm
performs $O(n^3\log(n\rho/\varepsilon))$ balancing operations and outputs a
matrix equivalent to~$A$ that is $\varepsilon$-balanced w.h.p.}

We give here a brief, high-level outline of our analysis.  First note that,
by Proposition~\ref{prop:twophasebasic} in the previous section, it suffices to analyze each 
phase of the algorithm separately.  Accordingly, we will analyze only the
raising phase, showing that it achieves an $\varepsilon$-raising-balanced
graph function w.h.p.\ after $O(n^3\log(\rho/\varepsilon))$ raising operations.
By symmetry the same bound holds for the lowering phase: simply observe
that lowering on the graph~$G_\alpha$ is equivalent to raising on the 
transpose graph $G^T_\alpha$.

In the Introduction we already discussed some of the difficulties which prevent us from analyzing the dynamics of the algorithm by tracking a scalar quantity at each vertex, as would be the case for the analogous but much simpler heat-kernel dynamics.  A key ingredient in our solution to the
problem is to introduce a two-dimensional representation of the
graph function~$\alpha$, together with an associated potential function that will enable
our analysis of the raising phase. One coordinate (the vertical coordinate) in this representation,
called the ``height" of a vertex, is based on the unique raising vector~$r^*$
whose existence is guaranteed by Lemma~\ref{lem:twophaseconv} of the previous section;
the second (horizontal) coordinate is the ``level" of the vertex, as alluded to in the Introduction.
This representation, which is specified in detail in Section~\ref{subsec:2drep},
will be instrumental in obtaining upper and lower bounds relating the (global)
potential function to the local imbalances (which govern the effect of a
single raising operation); these bounds are derived in Sections~\ref{subsec:prop1}
and~\ref{subsec:prop2} respectively.  Finally, in Section~\ref{subsec:final}
we conclude the proof of Theorem~\ref{thm:main}.

%%%%%%%%%%%%%%%%%
\subsection{A two-dimensional representation}\label{subsec:2drep} 
Throughout this and the next two subsections, we focus exclusively on the raising phase.
As we observed above, exactly the same rate of convergence will apply to the lowering phase.

Recall from Lemma~\ref{lem:twophaseconv} that, for any initial graph function~$\alpha$,
the limiting raising vector $r^*$, which records the asymptotic total amount of raising performed
at each vertex by any infinite fair sequence of raising operations, is uniquely defined and specifies
a unique raising balanced graph function~$\alpha^R$.
We may therefore use $r^*$ to define a {\it height function}  $$
   \y_v := -r^*_v, $$
i.e., $-\y_v$ is the amount of raising that remains to be performed at vertex~$v$ in 
order to reach the raising-balanced graph function~$\alpha^R$.  (Note that, throughout,
all heights are non-positive and converge monotonically to zero.  
This convention turns out to be convenient in our analysis.)  
The heights are related to $\alpha$ and~$\alpha^R$ by
\be \alpha^R_{uv}-\alpha_{uv}=\y_u-\y_v \label{ydiff}, \ee
which is just a restatement of equation~(\ref{eqn:alphar}).

In addition to the heights~$\y_v$, we will need to introduce a second coordinate~$\x_v$
defined by 
\be x_v := \max_{u:u \to v} \alpha^R_{uv}. \label{def:xR} \ee
We refer to 
$\x_v$ as the {\it level\/} of vertex~$v$.  Note that (intuitively at least), lower level
vertices can only reliably be close to raising-balanced after higher level ones are
close, so the levels capture one of the distinguishing features of the problem.
The two-dimensional representation $(x_v,\y_v)$ of vertices $v$ will be key to
the remainder of our argument.

Figure~\ref{fig:2dexample} shows the two-dimensional representation of the graph
function in Figure~\ref{fig:RLexample}(i). The horizontal coordinates can be read off
from the maximum incoming edges at each vertex in Figure~\ref{fig:RLexample}(ii), 
and the vertical coordinates from the square brackets in the same figure.

\begin{figure}[h]
\includegraphics[height=60mm]{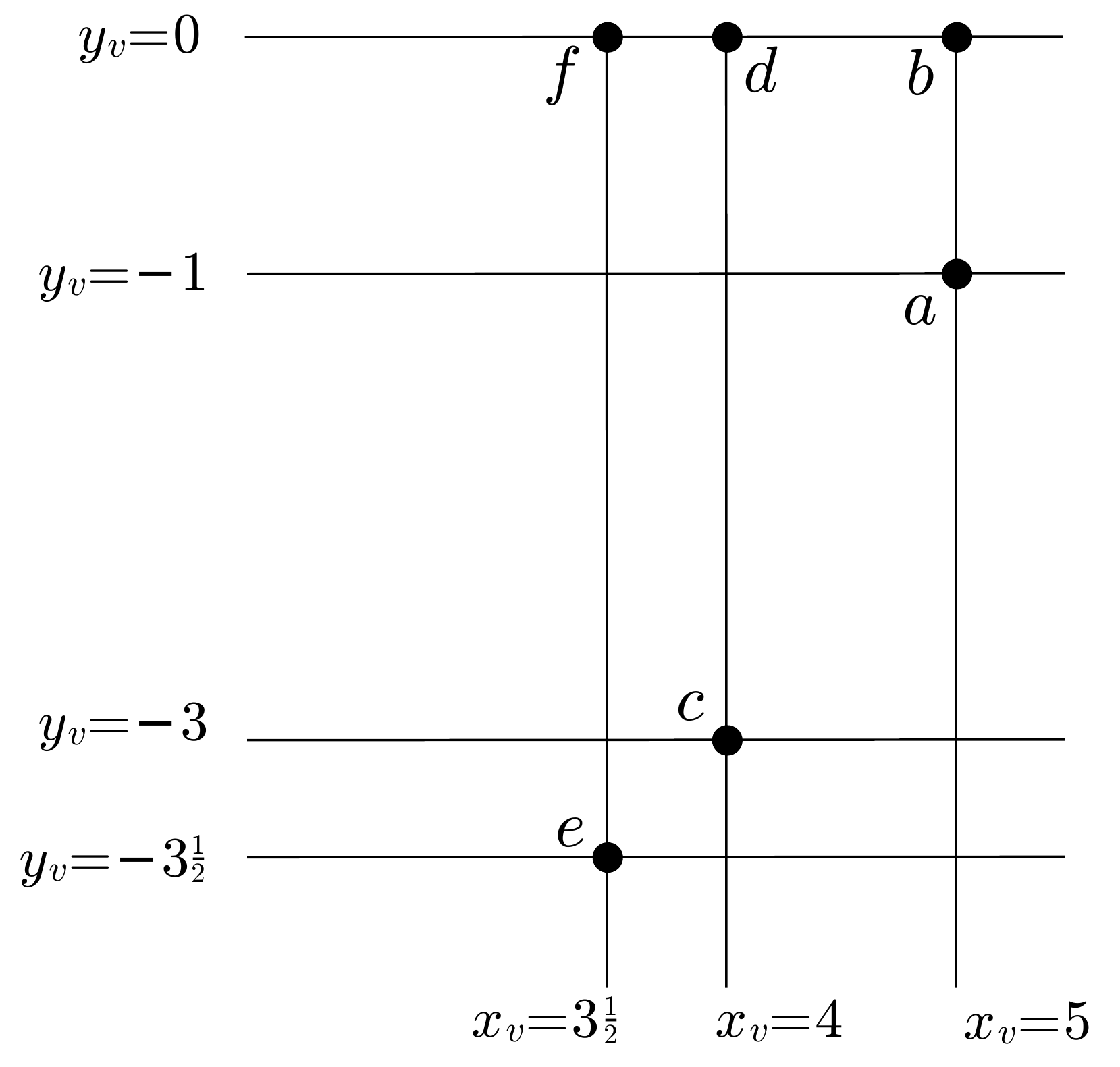}
\caption{\it Two-dimensional representation of the graph function~$\alpha$ shown
in Figure~\ref{fig:RLexample}(i).}
\label{fig:2dexample}
\end{figure}

We will measure progress of the raising phase by means of the global potential function
$\Psi=-\sum\y_v$, which is the total amount of raising that remains to be done to reach~$\alpha^R$.  
In order to relate $\Psi$ to the local imbalances~$\rho_v^R$, 
it will be convenient to introduce a closely related global quantity~$h$, which is just the
difference between the maximum and minimum heights. That is, setting
$\y_{\min}=\min_u \y_{u}$ and $\y_{\max}=\max_u \y_u$, we define $$
 h:=\y_{\max}-\y_{\min}.   $$
(As we shall see shortly (Corollary~\ref{cor:ymax0}), it is always the case that $\y_{\max}=0$
so in fact  $h=-\y_{\min}$.)  Note that $\Psi$ and $h$ differ by at most a factor
of~$n$, namely: \be \Psi\ge h\ge \frac{1}{n}\Psi. \label{psi-h} \ee

The key ingredients in our analysis are the following two bounds relating~$h$ to
the local imbalances~$\rho_v^R$. 

\begin{prop}\label{prop:h-lbR} 
For any graph function~$\alpha$, we have $\rho^R/2\le h$.
\end{prop}

\begin{prop}\label{prop:h-by-rho} 
For any graph function~$\alpha$, we have $h \le (n-1)\sum_v\rho^R_v$.
\end{prop}

Proposition~\ref{prop:h-lbR} confirms that, to ensure $\varepsilon$-raising-balance, it
suffices to achieve $\Psi\le\varepsilon/2$.  The much more subtle Proposition~\ref{prop:h-by-rho}
implies that~$\Psi$ decreases at each step by an expected $(1-\Omega(\frac{1}{n^3}))$ factor. To see this implication, note that 
the expected decrease in~$\Psi$ in one raising operation is $\frac{1}{2n}\sum_v\rho_v^R$, and apply Proposition~\ref{prop:h-by-rho} and equation~(\ref{psi-h}).

In the next two subsections we prove Propositions~\ref{prop:h-lbR} and~\ref{prop:h-by-rho}
respectively.

%%%%%%%%%%%%%%%%%
\subsection{Proof of Proposition~\ref{prop:h-lbR}}  \label{subsec:prop1}
The proposition will follow almost immediately from the following 
straightforward lemma relating heights to local imbalances. 
Since we are dealing exclusively with the raising phase, in this and the next
two subsections we will write $\rho_v,\rho$ in place of $\rho_v^R,\rho^R$ to simplify notation.

\begin{lemma} \label{y-uvw} Let $\alpha$ be a graph function and $v$ a vertex
with $\rho_v>0$.  Then there exist edges $(u,v)$ and $(v,w)$ in~$G_\alpha$ such that
\begin{equation}\label{eq:y-rhoR}
\y_v + \frac{\rho_v}{2} \le \frac{\y_{u}+\y_w}{2}.
\end{equation}
\end{lemma}

\begin{proof}
Let $u,w$ be vertices such that $\alpha_{uv}=\alpha^{\text{in}}_v$ and
$\alpha_{vw}=\alpha^{\text{out}}_v$ (i.e., $(u,v)$ and $(v,w)$ are incoming
and outgoing edges of maximum weight at~$v$).  Note that $\rho_v = \alpha_{vw}-\alpha_{uv}$.

To prove~(\ref{eq:y-rhoR}) we let $u'$ be a vertex such that the edge $(u',v)$ achieves
the maximum in~(\ref{def:xR}), i.e., $\alpha^R_{u'v} = x_v$,
and write
\begin{eqnarray}
   \y_v + \frac{\rho_v}{2} &=& \y_{u'} + \alpha_{u'v} - \alpha^R_{u'v} + \frac{\rho_v}{2}\nonumber\\
                          &=& \y_{u'} + \alpha_{u'v} - \alpha^R_{u'v} + \frac{\alpha_{vw} - \alpha_{uv}}{2} \nonumber\\
                          &\le& \y_{u'} - \alpha^R_{u'v}+ \frac{\alpha_{vw} + \alpha_{u'v}}{2},  \label{eq:tempR}
\end{eqnarray}
where in the first line we used equation~(\ref{ydiff}) and in the last line the fact 
that $\alpha_{u'v}\le\alpha_{uv}$.
Now two further applications of equation~(\ref{ydiff}) give
\begin{eqnarray*}
   \alpha_{vw} &=& \y_w - \y_v + \alpha^R_{vw}; \\
   \alpha_{u'v} &=& \y_v - \y_{u'} + \alpha^R_{u'v}.
\end{eqnarray*}
Plugging these into~(\ref{eq:tempR}) yields $$
   \y_v + \frac{\rho_v}{2} \le \y_{u'} - \alpha^R_{u'v} + \frac{\y_w - \y_{u'} + \alpha^R_{vw} + \alpha^R_{u'v}}{2} = \frac{\y_{u'}+\y_w}{2} + \frac{\alpha^R_{vw} - \alpha^R_{u'v}}{2} \le \frac{\y_{u'}+\y_w}{2}, $$
where the last inequality follows from the facts that $\alpha^R$ is raising balanced and 
$\alpha^R_{u'v}$ is maximal among incoming edges at~$u$.  This completes the proof. 
\end{proof}

Before proving Proposition~\ref{prop:h-lbR}, we pause to observe a consequence
of Lemma~\ref{y-uvw}:
\begin{corollary} \label{cor:ymax0}
$y_{\max}=0$.
\end{corollary}
\begin{proof} A raising operation applied at some vertex $v$ increases $y_v$ to 
$\y_v + \frac{\rho_v}{2}$. By Lemma~\ref{y-uvw} this is $ \le \frac{\y_{u}+\y_w}{2} \leq 
y_{\max}$. Therefore $y_{\max}$ can never increase during the raising phase. Choose a vertex $u$ so that $y_u=y_{\max}$ initially. Then $u$ can never rise, and so $r^*_u=0$. Thus $y_{\max}=y_u=-r^*_u=0$.
\end{proof}

We conclude this subsection by proving Proposition~\ref{prop:h-lbR}.
\par\medskip
\textit{Proof of Proposition~\ref{prop:h-lbR}.}
Let $v$ be any vertex with raising imbalance $\rho_v > 0$.
By Lemma~\ref{y-uvw} we have
$$ \frac{\rho_v}{2} \le  \frac{\y_{u} + \y_w}{2} - \y_v \le y_{\max}-y_{\min} = h.$$
Together with the fact that, by definition, $h\ge 0$, this completes the proof.\qed

%%%%%%%%%%%%%%%%%
\subsection{Proof of Proposition~\ref{prop:h-by-rho}}  \label{subsec:prop2}
The reverse inequality, Proposition~\ref{prop:h-by-rho}, is considerably more delicate. 
It will make crucial use of the two-dimensional representation introduced earlier,
as well as one additional concept which we call ``momentum," whose role will emerge
shortly.  The {\it momentum\/} of a vertex~$v$ is defined by
\be \label{def:momentumR} m_v := \max_{u:u \to v} \{\alpha_{uv}-\x_v\}, \ee
where $\x_v$ is the level of~$v$ as defined above.

One key property of momentum is that it can be upper bounded at a vertex in terms of the sum of the momentum and the local imbalance at the predecessors of the vertex in the graph. (The form of this bound is complicated by a dependence on the levels of the vertices; see part~(i) of Lemma~\ref{lemma:newlocalR}.) The second key property of momentum is that, along any edge, height can increase only at a rate bounded by momentum and local imbalance. (Again there are additional terms complicating the dependence; see part~(ii) of Lemma~\ref{lemma:newlocalR}.) 
Putting these two properties
together will show that, if the height difference between some two vertices in the graph is substantial,
then the aggregate local imbalance $\rhoone$ must be large, which is exactly the
content of Proposition~\ref{prop:h-by-rho}.

\begin{lemma}\label{lemma:newlocalR}
The momentum satisfies the following inequalities:
\begin{itemize}
\item[(i)]
For any vertex~$v$, \; $
   m_v \le \max\big\{0,\max_{u: u \to v, \y_u<\y_v} (m_u +\rho_u + \x_u - \x_v) \big\}$.
\vskip0.1in   
\item[(ii)]
For any edge $(u,v)$ in~$G_\alpha$, \; $\y_v \le \y_u + m_u + \rho_u + \x_u - \alpha^R_{uv}$.
\end{itemize}
\end{lemma}
\begin{proof}
Both parts rely on the following elementary observation.  For any edge
$(u,v)$, by the definition of $\rho_u$ we have
\begin{equation}\label{eqn:basicR}
   \alpha_{uv} \le \max_{u': u' \to u} \alpha_{u'u} + \rho_u
   = m_u + \rho_u + \x_u,
\end{equation}
where we have used the definition of momentum~(\ref{def:momentumR}).

For part~(i), we thus have $$
   m_v = \max_{u: u \to v}\{\alpha_{uv} - \x_v\} \le\max_{u: u \to v}\{m_u + \rho_u + \x_u -\x_v\}.  $$

Now note that if
$\y_u\ge \y_v$ we have, recalling equations~(\ref{ydiff},~\ref{def:xR}), \be \label{u-above-vR}
   \alpha_{uv} - \x_v = \y_v - \y_u + \alpha^R_{uv} - \x_v \le 0. \ee
This justifies the restriction on the right-hand side of (i) to maximization over 
 $u$ s.t.~$\y_u<\y_v$.
This completes the proof of~(i).

For part~(ii), apply equations~(\ref{ydiff},~\ref{eqn:basicR}) to deduce $$
   \y_v = \y_u + \alpha_{uv} - \alpha^R_{uv} \le \y_u + m_u + \rho_u + \x_u - \alpha^R_{uv}, $$
as required.   
\end{proof}

To complete the proof of Proposition~\ref{prop:h-by-rho}, we need to
translate Lemma~\ref{lemma:newlocalR}, which gives local bounds on 
the increase of momentum and height between adjacent vertices,
into global bounds in~$G_\alpha$, as specified in Lemma~\ref{lem:newglobalR} below.  
The key to this translation will be to apply part~(ii) of Lemma~\ref{lemma:newlocalR} to a
carefully chosen sequence of edges, for each of which the extra term
$x_u-\alpha^R_{uv}$ in the lemma is non-positive and which together span
the heights in the graph; these edges will be identified in Lemma~\ref{lem:technical} below.

For a vertex~$v$, let $S_v$ denote the
set of vertices whose height is strictly smaller than~$v$, i.e., $S_v = \{u: \y_u < \y_v\}$.
Also, let $\rho(S_v) = \sum_{u\in S_v} \rho_u$.
Finally, define
 \be z_v := \max\{\x_u: \y_u \leq \y_v\}-\x_v. \label{def-z} \ee
This definition has the consequences that $z_v \geq 0$ and
 \be \label{xz-monotoneR} [y_u \leq y_v] \Rightarrow [x_u+z_u\leq x_v+z_v]. \ee

We are now ready to state our global bounds. 
\begin{lemma}\label{lem:newglobalR}
For any vertex~$v$, the following two bounds hold:
\begin{enumerate}
\item[(i)] $m_v \le \rho(S_v) + \z_v$\/{\rm ;}
\vskip0.1in
\item[(ii)] $\y_v - y_{\min} \le |S_v| \cdot \rho(S_v)$.
\end{enumerate}
\end{lemma}

\begin{proof}
Both parts~(i) and~(ii) are proved by induction on vertices~$v$ in order of height.  
We begin with part~(i), which will be used in the proof of part~(ii).

{\it Part~(i)}.  We use induction on the height~$\y_v$ of~$v$.  If $\y_v=\min_u y_u$, the statement is true since $m_v \leq 0$ (this follows from
part~(i) of Lemma~\ref{lemma:newlocalR})
and the right-hand side is non-negative.  Assuming $\y_v>\min_u y_u$, and beginning
with part~(i) of Lemma~\ref{lemma:newlocalR}, we may write
\begin{eqnarray*}
   m_v &\le& \max\big\{0,\max_{u: u \to v, \y_u<\y_v}\{ m_u +\rho_u + \x_u - \x_v\}\big\}\\
            &\le& \max\big\{0,\max_{u: u \to v, \y_u<\y_v}\{ \rho(S_u) +\z_u + \rho_u + \x_u - \x_v\}\big\}\\
            &\le& \max\big\{0,\max_{u: u \to v, \y_u<\y_v}\{ \rho(S_v)  +\z_u  + \x_u - \x_v\}\big\},
\end{eqnarray*}
where in the second line we applied induction to bound~$m_u$ since $\y_u < \y_v$,
and in the third line we used the fact that $u\in S_v\setminus S_u$ to deduce that
$\rho(S_u) + \rho_u\le\rho(S_v)$.

To complete the proof we need only apply Eqn.~(\ref{xz-monotoneR}) to conclude that for all $u$ in the last maximization, $x_u+z_u \leq x_v+z_v$.

{\it Part~(ii)}.  Again we use induction on the height of~$v$.  The base case
is $\y_v=y_{\min}$, in which case the statement holds trivially.

Now let $v$ be any vertex with $y_{\min}<\y_v$, and suppose part~(ii) is established for all $u$ with $y_u<\y_v$. Let
$U$ be the subset of $\{u: y_u<y_v\}$ having $x_u=x_U$, where $x_U=\max\{x_u: y_u<y_v\}$. Of course $U\neq \emptyset$, and $z_u=0$ for all $u \in U$.  We require the following additional fact:

\begin{lemma}\label{lem:technical} 
There is a vertex $u' \in U$ having an edge $(u',v')$ where $v' \notin U$, 
$\alpha^R_{u'v'}\geq x_U$, 
 $y_{v'}\geq y_v$.
\end{lemma}

We defer the proof of Lemma~\ref{lem:technical} to the end of the subsection and
continue with the proof of Lemma~\ref{lem:newglobalR}(ii). 
Applying part~(ii) of Lemma~\ref{lemma:newlocalR} to the edge $(u',v')$ provided by 
Lemma~\ref{lem:technical}, we have
\begin{eqnarray*}
   \y_v\,\le\, \y_{v'} &\le& \y_{u'} + m_{u'} + \rho_{u'} + \x_{u'} - \alpha^R_{u'v'}\\
                          &\leq & \y_{u'} + m_{u'} + \rho_{u'}\\
                          &\le& \y_{u'} + \rho(S_{u'}) + \z_{u'} + \rho_{u'}\\
 & = &  \y_{u'} + \rho(S_{u'}) + \rho_{u'}\\
                          &\le& y_{\min}+|S_{u'}|\rho(S_{u'}) + \rho(S_{u'}) + \rho_{u'}\\
                          &\le& y_{\min}+|S_{v}|\rho(S_{v}).
\end{eqnarray*}
In the second line here we have used the guarantee of Lemma~\ref{lem:technical} that 
$\alpha^R_{u'v'}\geq x_U$;
in the third line we have used part~(i) of the
current lemma;
in the fourth line we have used the fact noted earlier that $z_{u'}=0$ for any $u'\in U$;
in the fifth line we have used the inductive
hypothesis applied to~$\y_{u'}$ (which is strictly less than~$\y_v$); and in the last line
we have used the fact that $u'\in S_{v}\setminus S_{u'}$.
This completes the inductive proof of part~(ii) of the lemma.
\end{proof}

Our main goal in this subsection, Proposition~\ref{prop:h-by-rho}, now follows trivially from
Lemma~\ref{lem:newglobalR}.
\par\medskip
\textit{Proof of Proposition~\ref{prop:h-by-rho}.}
By part~(ii) of Lemma~\ref{lem:newglobalR} we have, for any vertex~$v$, 
$\y_v-\y_{\min} \le (n-1)\sum_u\rho_u$.\qed
\par\medskip
It remains only for us to supply the missing proof of Lemma~\ref{lem:technical}.
\par\medskip
\textit{Proof of Lemma~\ref{lem:technical}.}
The proof will follow from two claims.

\textit{Claim (i):} Every vertex $u\in U$ has an outgoing edge $(u,w)$ with $\alpha^R_{uw}= x_U$. 

To state the second claim, define the graph $G'$ formed by the set of edges 
$\{(u,w): u\in U, \alpha^R_{uw}\ge x_U\}$.

\textit{Claim (ii):} There is an edge $(u,w) \in G'$ such that $w \notin U$.

Claim (ii) implies the lemma because $(u,w) \in G'$ implies $\alpha^R_{uw}\ge x_U$; this last implies also that $x_w \geq x_U$; and then $w$ cannot have $y_w<y_v$ as it would then be in $U$.

\textit{Proof of Claim (i):} Since $u\in U$, we know that $\y_u<0$ and hence some raising must
eventually be performed at~$u$, i.e., eventually $\rho^R_u>0$ and then $\rho^L_u=0$.   
But once this happens, the
same argument as in the proof of Proposition~\ref{prop:twophasebasic}
shows that $\rho^L_u=0$ at all future times, and hence in~$\alpha^R$ (at the conclusion of the raising phase) we must
have $\rho^R_u=\rho^L_u=0$.  Thus if $(u,w)$ is a maximum outgoing edge at~$u$ in~$\alpha^R$,
we have $\alpha^R_{uw}=\max_v \alpha^R_{vu} = x_U$.\qed

\textit{Proof of Claim (ii):} By Claim~(i), every $u\in U$ has an outgoing edge in $G'$. 
Suppose for contradiction that $w\in U$ for all edges of $G'$.
Then in $G'$ there is a sink scc on vertices $U' \subseteq U$. 
($U'$~may be as small as a single vertex with a self-loop.) 

Let $\delta_0 = -\max_{u'\in U'} \y_{u'}$.  Note that $\delta_0 > 0$.  
Also, let $\delta_1= \min\{\alpha^R_{rr'}-\alpha^R_{ss'}: \alpha^R_{rr'}-\alpha^R_{ss'}>0\}$;
if this set is empty let $\delta_1=\infty$.  Finally, select $0<\delta<\min\{\delta_0,\delta_1\}$.

Now fix any infinite fair sequence of raising operations, and consider the first operation
in the sequence that increases any $\y_{u}$, $u\in U'$, to a value larger than $-\delta$. 
(This is well defined as all $\y_{u}, u \in U'$ increase monotonically to zero.) 
Let $u'$ be the vertex raised by this operation, and let $\alpha$ be
the graph function immediately after the operation and $\y$ the 
height function corresponding to~$\alpha$.

By construction there is an edge $(t,u')$ in $G'$, with $t \in U'$. Applying 
equation~(\ref{ydiff}) to the edge $(t,u')$, we see that 
\begin{equation}\label{contra0}
  \alpha_{tu'} = \alpha^R_{tu'} + y_{u'} - y_t \ge x_U + y_{u'} - y_t  \geq x_U,
\end{equation}
where the first inequality is because the edge is in $G'$, and the second inequality
is by choice of~$u'$.  (Note that equality can hold here only if $t=u'$.)

Now let $(u',w)$ be a maximum (in $\alpha$) outgoing edge at~$u'$. 
Since a nontrivial raising operation
has just been performed at~$u'$, it must be the case that $w\ne u'$. (Raising preserves the 
order of weights among outgoing edges in~$\alpha$. If a self-loop is a maximal outgoing edge from $u'$ in $\alpha$ then it was also maximal before raising. But then there could have been no nontrivial raising at~$u'$.) Also, since we have just raised at $u'$, $\rho^L_{u'}=0$, so $\alpha_{u'w}\ge\alpha_{tu'}$. Combining with~(\ref{contra0}) we have
\begin{equation}\label{contra1}
       \alpha_{u'w}\ge x_U.
\end{equation}     

Applying equation~(\ref{ydiff}) to the edge $(u',w)$ we get 
$$ \alpha^R_{u'w} = \alpha_{u'w} + \y_{u'} - \y_w \geq \alpha_{u'w} + \y_{u'} \geq 
\alpha_{u'w} - \delta \geq
 x_U - \delta,  $$ 
where the first inequality is because heights are nonpositive, 
the second is by the choice of $u'$, and the third is by equation~(\ref{contra1}).
But since $\delta < \delta_1$ this in fact implies that $\alpha^R_{u'w} \geq x_U$.
Hence the edge $(u',w)$ belongs to~$G'$, and since $U'$ is a sink scc
it must be the case that $w\in U'$.  But this implies that  $\alpha^R_{u'w}\le x_w=x_U$, and therefore in fact that  $\alpha^R_{u'w} = \x_U$; and also, by choice of~$u'$, that
$\y_w<\y_{u'}$. Hence $$
   \alpha_{u'w} = \alpha^R_{u'w} -\y_{u'} + \y_w = x_U  -\y_{u'} + \y_w
 < x_U, $$
which in light of~(\ref{contra1}) gives us the desired contradiction.
This completes the proof of Claim~(ii) and therefore the lemma. \qed

%%%%%%%%%%%%%%%%%
\subsection{Proof of Theorem~\ref{thm:main}}  \label{subsec:final}
We are now finally in a position to prove our main result, Theorem~\ref{thm:main}.
\par\medskip
\textit{Proof of Theorem~\ref{thm:main}.}
It suffices to show that, with high probability, $O(n^3 \log (\rho n/\ep))$ raising operations
suffice to achieve an $\varepsilon$-raising-balanced graph function.  As we observed
at the beginning of the section, by symmetry the same holds for
the lowering phase, and hence the output of the algorithm will be $\varepsilon$-balanced
w.h.p.

To analyze the raising phase, we index all quantities by~$t$, the number of 
raising operations that have so far been performed.  Thus in particular 
$\y_v(t)$ is the height of vertex~$v$ after $t$ raising operations.
We also introduce the potential function $\Psi(t) = -\sum_v \y_v(t)$.
Note that $\Psi(t)\ge 0$ and that $\Psi(t)$ decreases monotonically to~0 as $t\to\infty$.  
Moreover, by Proposition~\ref{prop:h-lbR} we have $$
   \Psi(t) = -\sum_v \y_v(t) \ge h(t) \ge \rho^R(t)/2.  $$
Hence in order to achieve raising imbalance at most~$\varepsilon$ after $t$ raising
steps it suffices to ensure that $\Psi(t)\le \varepsilon/2$.

Since a raising operation at vertex~$v$ increases $\y_v$ by $\rho^R_v/2$, the
expected decrease in $\Psi(t)$ in one raising operation is 
\begin{equation*}
    {\rm E}[\Psi(t)-\Psi(t+1)] = \frac{1}{2n}\sum_v \rho^R_v(t)
                                                \ge \frac{h(t)}{2n(n-1)} 
                                                = \frac{-\y_{\min}(t)}{2n(n-1)}
                                                \ge \frac{1}{2n^3}\Psi(t),
\end{equation*}
where in the first inequality we have used Proposition~\ref{prop:h-by-rho}, in the
next step Corollary~\ref{cor:ymax0}, and in the last inequality the fact that $\y_{\min}$ 
is less than the average of the $\y_v$.
Iterating for $t$ steps yields $$
    {\rm E}[\Psi(t)] \le \Bigl(1-\frac{1}{2n^3}\Bigr)^t \Psi(0).  $$
To get an upper bound on $\Psi(0)$, we observe that $$
   \Psi(0) = -\sum_v \y_v(0) \le -n\y_{\min}(0) = nh(0) \le n(n-1)\sum_v\rho^R_v(0) \le n^3\rho^R(0) \le n^3\rho,  $$
where we have again used Corollary~\ref{cor:ymax0} and Proposition~\ref{prop:h-by-rho}.  (Recall
that $\rho=\max\{\rho^R(0),\rho^L(0)\}$ is the imbalance of the original matrix.)
Hence, for any $\delta\in(0,1)$, after $t=6n^3 \ln (2\rho n/\ep\delta)$ raising operations we have $$
    {\rm E}[\Psi(t)] \le \varepsilon\delta/2.  $$
By Markov's inequality this implies that $\Psi(t)\le\ep/2$ with probability at least $1-\delta$.
To obtain the result w.h.p., it suffices to take $\delta=o(1)$.
\qed

%%%%%%%%%%%%%%%%%%%%%%%
\section{Open problems}
We close the paper with some remarks and open problems.

{\bf 1.}\ \ The introduction of raising and lowering phases into the Osborne-Parlett-Reinsch
algorithm is trivial from an implementation perspective (the effect being simply to censor
some of the balancing operations), but seemingly crucial to our analysis.  If balancing
operations of both types are interleaved, we no longer have a reliable measure of
progress.  It would be interesting to modify the analysis so as to handle
an arbitrary sequence of balancing operations, and on the other hand to investigate
whether the introduction of phases actually improves the rate of convergence (both in
theory and in practice).  Similarly, it would be interesting to investigate the effect of
choosing a random sequence of indices at which to perform balancing, rather than
a deterministic sequence as in the original algorithm.

{\bf 2.}\ \ Recall that, without raising and lowering phases, the balanced matrix to which
the algorithm converges is not unique (except, by definition, in the UB case).  It would
be interesting to investigate the geometric structure of the set of fixed points, which apparently
can be rather complicated.

{\bf 3.}\ \ As we have seen, our $\Otilde(n^3)$ bound on the worst case convergence time is
essentially optimal.  While this goes some way towards explaining
the excellent performance of the algorithm in practice, more work needs to be done to explain its empirical dominance over algorithms with faster worst-case asymptotic running times.
In particular, can one identify features of ``typical'' matrices that lead to much
faster convergence?

{\bf 4.}\ \ We have analyzed only the $L_\infty$ variant of the Osborne-Parlett-Reinsch
algorithm.  It would be interesting to analyze the algorithm for other norms $L_p$, in
particular $L_2$ and $L_1$.  (In the latter case, the problem corresponds to finding
a circulation---in the sense of network flows---that is equivalent to the initial assignment
of flows to the edges of the graph.)

\section{Acknowledgements}
We thank Tzu-Yi Chen and Jim Demmel for introducing us to this problem and for several helpful
pointers, Martin Dyer for useful discussions, and Yuval Rabani for help in the 
early stages of this work which led to the proof of Theorem~\ref{thm:characterizeUB}.  
We also thank anonymous referees of an earlier version for suggestions that improved
the presentation of the paper. 

\bibliography{refs} \bibliographystyle{plain}

\begin{thebibliography}{10}

\bibitem{Boyd94}
S.~Boyd, L.~El Ghaoui, E.~Feron, and V.~Balakrishnan.
\newblock {\em Linear Matrix Inequalities in System and Control Theory}.
\newblock SIAM Studies in Applied Mathematics Vol.~15, 1994.

\bibitem{Chen-MS98}
T.-Y. Chen.
\newblock Balancing sparse matrices for computing eigenvalues.
\newblock Master's thesis, UC Berkeley, May 1998.

\bibitem{ChenD00}
T.-Y. Chen and J.~Demmel.
\newblock Balancing sparse matrices for computing eigenvalues.
\newblock {\em Linear Algebra and its Applications}, 309:261--287, 2000.

\bibitem{EavesHRS85}
B.~C. Eaves, A.~J. Hoffman, U.~G. Rothblum, and H.~Schneider.
\newblock Line-sum-symmetric scalings of square non-negative matrices.
\newblock {\em Math. Programm. Study}, 25:124--141, 1985.

\bibitem{FL89}
J.~Franklin and J.~Lorenz.
\newblock On the scaling of multidimensional matrices.
\newblock {\em Linear and Multilinear Algebra}, 23:717--735, 1989.

\bibitem{Grad71}
J.~Grad.
\newblock Matrix balancing.
\newblock {\em The Computer Journal}, 14:280--284, 1971.

\bibitem{Hartfiel71}
D.~J. Hartfiel.
\newblock Concerning diagonal similarity of irreducible matrices.
\newblock {\em Proceedings of the American Mathematical Society}, 30:419--425,
  1971.

\bibitem{KalKhach93}
B.~Kalantari and L.~Khachiyan.
\newblock {On the rate of convergence of deterministic and randomized RAS
  matrix scaling algorithms}.
\newblock {\em Operations Research Letters}, 14:237--244, 1993.

\bibitem{KalKhach96}
B.~Kalantari and L.~Khachiyan.
\newblock On the complexity of nonnegative matrix scaling.
\newblock {\em Linear Algebra and its Applications}, 240:87--104, 1996.

\bibitem{Kalantari97}
B.~Kalantari, L.~Khachiyan, and A.~Shokoufandeh.
\newblock On the complexity of matrix balancing.
\newblock {\em SIAM Journal on Matrix Analysis and Applications}, 18:450--463,
  1997.

\bibitem{KLRS08}
B.~Kalantari, I.~Lari, F.~Ricca, and B.~Simeone.
\newblock {On the complexity of general matrix scaling and entropy minimization
  via the RAS algorithm}.
\newblock {\em Mathematical Programming}, 112(2):371--401, 2008.

\bibitem{Kressner05}
D.~Kressner.
\newblock {\em Numerical Methods for General and Structured Eigenvalue
  Problems}.
\newblock Springer, 2005.

\bibitem{LSW00}
N.~Linial, A.~Samorodnitsky, and A.~Wigderson.
\newblock A deterministic strongly polynomial algorithm for matrix scaling and
  approximate permanents.
\newblock {\em Combinatorica}, 20:531--544, 2000.

\bibitem{Osborne60}
E.~E. Osborne.
\newblock On pre-conditioning of matrices.
\newblock {\em Journal of the ACM}, 7(4):338--345, 1960.

\bibitem{ParlettR69}
B.~N. Parlett and C.~Reinsch.
\newblock Balancing a matrix for calculation of eigenvalues and eigenvectors.
\newblock {\em Numerische Mathematik}, 13:293--304, 1969.

\bibitem{SchneiderS91}
H.~Schneider and M.~H. Schneider.
\newblock Max-balancing weighted directed graphs and matrix scaling.
\newblock {\em Mathematics of Operations Research}, 16(1):208--222, 1991.

\bibitem{SSSTOC}
L.~J. Schulman and A.~Sinclair.
\newblock Analysis of a classical matrix preconditioning algorithm.
\newblock {\em Proceedings of the 47th Annual ACM Symposium on Theory of
  Computing (STOC)}, 2015.

\bibitem{Sinkhorn64}
R.~Sinkhorn.
\newblock A relationship between arbitrary positive matrices and doubly
  stochastic matrices.
\newblock {\em Annals of Mathematical Statistics}, 35:876--879, 1979.

\bibitem{Strom72}
T.~Str\protect{\"o}m.
\newblock Minimization of norms and logarithmic norms by diagonal similarities.
\newblock {\em Computing}, 10:1--7, 1972.

\bibitem{TrefethenEmbree05}
L.~N. Trefethen and M.~Embree.
\newblock {\em Spectra and Pseudospectra: The Behavior of Nonnormal Matrices
  and Operators}.
\newblock Princeton University Press, 2005.

\bibitem{YoungTarjanOrlin91}
N.~Young, R.~E. Tarjan, and J.~B. Orlin.
\newblock Faster parametric shortest path and minimum balance algorithms.
\newblock {\em Networks}, 21:205--221, 1991.

\end{thebibliography}
\appendix

%%%%%%%%%%
%%%%% yi-reproof
%%%%%%%%%%
\section{Proof of convergence} \label{app:yi-reproof}
In this section we present an alternative proof of convergence of the original
Osborne-Parlett-Reinsch $L_\infty$ balancing algorithm.  As remarked earlier
this was first proved by
Chen~\cite{Chen-MS98} using a compactness argument.  In the interests of making the
present paper self-contained, we give here an alternative proof that is more combinatorial
in nature.  For this proof only, we explicitly allow more general sequences of balancing
operations, subject only to the ``fairness" constraint that every index~$i$ appears infinitely often.
(Clearly, the random sequence discussed in the main text has this property with probability~1.)
Note, however, that this proof does not apply to the two-phase version of the algorithm
which we analyze in this paper because the sequence of balancing operations in that
version is not fair.  We prove convergence of the two-phase version in Section~\ref{sec:convergence}.

\begin{theorem}\label{thm:chen}
Consider the Osborne-Parlett-Reinsch $L_\infty$ balancing algorithm applied
to any irreducible (real or complex) input matrix~$A$, and assume that balancing
operations are performed infinitely often at all indices~$i$.  Then the algorithm
converges, i.e., if $A^{(t)}$ is the sequence of matrices after each step of the
algorithm, then there exists a balanced matrix~$B$ equivalent to~$A$ such that
$A^{(t)}\to B$ as $t\to\infty$.
\end{theorem}
{\it Proof.}
We work in
the framework of Section~\ref{sec:prelim}, representing~$A$ as a graph
$G_\alpha$ with an associated graph function~$\alpha$, where
$\alpha_{uv}=\log |a_{uv}|$ (and edges $(u,v)$ with $a_{uv}=0$ are omitted).
We let $\alphavin$ and $\alphavout$ denote the maximum incoming
and outgoing edge weights at vertex~$v$ in $G_\alpha$; we also set
$\alphavmax=\max\{\alphavin,\alphavout\}$.

The key ingredient in the proof is the following partial order on graph functions
defined on a common strongly-connected graph~$G$.  For two such graph functions
$\alpha, \gamma$, we say that $\alpha<\gamma$ if, for the largest weight~$w$
such that $G_\alpha^w \neq G_\gamma^w$, we have $G_\alpha^w \subset G_\gamma^w$.

Now, if a decreasing sequence~$\alpha^{(t)}$ of graph functions
in this order has the property that all values $\alpha^{(t)}_{uv}$ in the functions
are bounded below by a common value~$b$, it is clear that the sequence has a limit.
The convergence of the iterative balancing
algorithm will therefore follow from the next two lemmas.  Call a balancing operation
at~$i$ {\it non-trivial\/} if $\alphavin\ne\alphavout$, so that the operation actually changes
the weights.  By definition $\alpha$ is balanced iff no non-trivial balancing operation is available.

\begin{lemma} Non-trivial balancing operations are strictly decreasing in the above order on
graph functions.
\end{lemma}
\begin{proof}
Let $\alpha$ be the current graph function, and let~$v$ be a vertex at which
a non-trivial balancing operation is performed.  Let $\alphabar$ denote the new
graph function after the operation.
We have $\alphabarvmax=\alphavmax-|\alphabarvin-\alphabarvout|/2$, and
$\alphabarvmax<\alphavmax$.
The balancing operation affects only edges incident at~$v$. Before the balancing, at least one of those edges belongs to $G_\alpha^{\alphavmax}$, while none belong to $G_\alpha^w$ for any $w>\alphavmax$. After the balancing, none of those edges belong to
$G_\alpha^{\alphavmax}$.
Thus the largest value of~$w$ for which $G_\alpha^w$
and $G_{\alphabar}^{w}$ differ is $w=\alphavmax$, and
$G_{\alphabar}^{\alphavmax}\subset G_\alpha^{\alphavmax}$.
\end{proof}

\begin{lemma}\label{lem:bdedbelow}
Given any graph function $\alpha$, there is a
$b>-\infty$ such that no sequence of balancing operations can
create an edge of weight less than~$b$.
\end{lemma}
\begin{proof} A balancing operation at vertex~$v$
reduces $\alphavmax$, so no such operation can increase
the maximum edge value above its original value in~$G_\alpha$,
say~$m_1$. Let~$w_1$ be the least average weight of any
cycle in~$G_\alpha$; note that $w_1$ remains invariant under any sequence
of balancing operations.  Consider the value of any edge $(u,v)$ under any graph
function $\alpha'$ that is reached by balancing operations from~$\alpha$.
Let $(u,v)$ belong to some (simple) cycle of average weight~$w_2$ and length~$\ell$.
Then $(\alpha'_{uv}+(\ell-1)m_1)/\ell \geq w_2 \geq w_1$, so
$\alpha'_{uv} \geq w_1-(\ell-1)(m_1-w_1) \geq w_1-(n-1)(m_1-w_1)$.
This completes the proof of the lemma and the theorem.
\end{proof}

\end{document}